\documentclass{article}
\usepackage{amsmath,amsthm,amssymb}
\usepackage{array, longtable}
\newtheorem{Theorem}{Theorem}[section]
\newtheorem{lem}[Theorem]{Lemma}
\newtheorem{Remark}[Theorem]{Remark}
\newtheorem{Definition}[Theorem]{Definition}

\newtheorem{Proposition}[Theorem]{Proposition}

\numberwithin{equation}{section}
\usepackage[latin1]{inputenc}

\begin{document}

\title{Improved Singleton bound on insertion-deletion codes and optimal constructions\footnote{
 {\small Email addresses: bocong\_chen@yahoo.com (B. Chen), zghui@lynu.edu.cn (G. Zhang).}}}

\author{Bocong Chen$^1$, Guanghui Zhang$^2$}

\date{\small
${}^1$School of Mathematics, South China University of Technology, Guangzhou 51064, China\\
${}^2$School of Mathematical Sciences,
Luoyang Normal University,
Luoyang, Henan, 471934, China
}

\maketitle
\begin{abstract}
Insertion-deletion codes (insdel codes for short) play an important role in synchronization error correction.
The higher the minimum insdel distance,
the more insdel errors the code can correct.
Haeupler and   Shahrasbi established the Singleton bound for insdel codes:
the minimum insdel distance of any $[n,k]$ linear code over $\mathbb{F}_q$ satisfies   $d\leq2n-2k+2.$
There have been some constructions of insdel codes through Reed-Solomon codes with high capabilities, but none has come close to this bound. Recently, Do Duc {\it et al.} showed that the minimum insdel distance of any $[n,k]$ Reed-Solomon code is no more than $2n-2k$ if $q$ is large enough compared to the code length $n$;  optimal codes that meet the new bound were also constructed explicitly.
The contribution of this paper is twofold.
We first show that the minimum insdel distance of any $[n,k]$ linear code over $\mathbb{F}_q$ satisfies $d\leq2n-2k$ if $n>k>1$.
This result  improves and generalizes  the previously known  results in the literature.
We then give a sufficient condition under which the minimum insdel distance of a two-dimensional
Reed-Solomon code  of length $n$ over $\mathbb{F}_q$ is exactly equal to $2n-4$.
As a consequence, we show that the sufficient condition is  not hard to achieve;
we explicitly construct an infinite family of optimal two-dimensional Reed-Somolom codes meeting the bound.


\end{abstract}

\section{Introduction}
Insertion-deletion codes (insdel codes for short) are   designed to protect against synchronization errors \cite{hs}, \cite{hss} in communication systems
caused by the loss of positional information of the message.
Insdel codes have found applications in many interesting fields such as DNA storage, DNA analysis \cite{jhsb}, \cite{xw}, language processing \cite{bm}, \cite{och}
and race-track memory error correction \cite{ckvvy}.

The insdel distance between two vectors is defined as the smallest number of insertions and deletions needed to transform
one codeword into another. The minimum insdel distance of a code is defined in the natural way: the minimum insdel distance among all
its distinct codewords.
Like the classical linear codes with respect to the Hamming distance,
the minimum insdel distance of an insdel  code is an important parameter,
which shows  its insdel error-correcting capability.
The higher the minimum insdel distance,
the more insdel errors the code can correct.
The study of insdel codes can be date back to the 1960s \cite{vt}. Insdel codes with   small code lengths were constructed explicitly by using various mathematical methods, e.g., see \cite{bo}, \cite{ma}, \cite{yin}. Sloane \cite{slo} constructed
a family of codes capable of correcting single deletion.

For a fixed code length, it would certainly be nice if
both the code size   (which is a measure of the efficiency
of the code) and the minimum insdel distance   could be
as large as possible. However, as in the Hamming metric
case, these two parameters are restricted each other for any
fixed code length. The Singleton  bound for insdel codes
says that the minimum insdel distance of
any $[n,k]$ linear code over $\mathbb{F}_q$  satisfies  $d\leq2n-2k+2$, see \cite{hs}.
It is, therefore,  natural to consider the problem of constructing  insdel linear codes achieving the Singleton bound with equality. Unfortunately,
there have been few constructions of Reed-Solomon codes with high capabilities, but none of them  meets or comes close to this bound.
For example, Wang {\it et al.} \cite{wms} constructed a class of Reed-Solomon codes with code length $n=5$, dimension $k=2$ and minimum insdel distance $d$ at most $4$;
Tonien \cite{tsn} {\it et al.} constructed a class of
generalized Reed-Solomon codes of length $n$ and dimension $k$ with deletion error-correcting capability of up to $\log_{k+1}n-1$.
These led  to investigate whether the Singleton bound is a tight upper bound for the minimum insdel distance of Reed-Solomon codes.
In 2007, McAven {\it et al.} \cite{mcs}
showed that Reed-Solomon codes of length   $n\geq 3$ and dimension $2$ over prime fields can never meet the Singleton bound.
Recently, Do Duc {\it et al.} \cite{dltx} improved the result by showing that when the field size is sufficiently large compared to the code length,
the Singleton bound cannot be achieved by Reed-Solomon codes; more explicitly, it was shown that the minimum
insdel distance of a $k$-dimensional Reed-Solomon code is at most $2n-2k$ if the code length $n$ satisfies
$n>k>1$ and $q>n^2$ (\cite[Theorem 1]{dltx});
optimal codes that meet the new bound were also constructed explicitly \cite[Theorems 2 and 3]{dltx}.
Very recently, Liu {\it et al.} in \cite{lt}   established a set of sufficient conditions for two-dimensional insdel Reed-Solomon codes to have
optimal asymptotic error-correcting capabilities.

The aforementioned works lead us to the study of Singleton-type
bounds for the minimum insdel distance of general linear codes
and optimal constructions of such  codes. The contribution
of this paper is twofold. We first show that
the minimum insdel distance of any $[n,k]$ linear code over $\mathbb{F}_q$ satisfies  $d\leq2n-2k$ if $n>k>1$.
This result  improves and generalizes   \cite[Theorem 1]{dltx} in two directions:
First, our result holds true for general linear codes, not just Reed-Solomon codes; second, we do not require that $q>n^2$.
More precisely, we obtain the following result.

\vskip 0.1cm
\textbf{Theorem A}~
{\it Suppose $\mathcal{C}$ is an $[n,k]$ linear code over $\mathbb{F}_q$ with $n>k\geq2$.
Then the minimum insdel distance of $\mathcal{C}$ is at most $2n-2k$, i.e.,  $d(\mathcal{C})\leq2n-2k$.}
\vskip 0.1cm

We then give   a sufficient condition under which the minimum insdel distance of  a  two-dimensional
Reed-Solomon code  of length $n$ over $\mathbb{F}_q$   is exactly equal to  $2n-4$. Our approach and conclusion are quite different from those
given in \cite[Theorems 2 and 3]{dltx}: the proofs for \cite[Theorems 2 and 3]{dltx} are long and technical,   and \cite[Theorems 3]{dltx} requires
 conditions on the divisors of $q-1$ and some related values;
our methods are more natural and direct, and our result mainly concerns the order of the finite field.
Our result is stated below.

\vskip 0.1cm
\textbf{Theorem B}~
{\it Let $\mathbb{F}_q$ be a finite field with $q=p^e$ elements, where $p$ is a prime number and $e>1$
is a positive integer.
Suppose $\theta$ is a primitive element of $\mathbb{F}_q$, i.e.,
the order of $\theta$ in the multiplicative group of $\mathbb{F}_q$ is equal to $q-1$.
Let
$$\mathcal{C}=\Big\{\Big(\lambda+\mu\theta^{i_1}, \lambda+\mu\theta^{i_2},\cdots,\lambda+\mu\theta^{i_n}\Big)\,\Big{|}\,\lambda,\mu\in \mathbb{F}_q\Big\}$$
be a two-dimensional Reed-Solomon code of length $n$ over $\mathbb{F}_q$. We assume that $0\leq i_1<i_2<\cdots<i_n\leq q-2$.
If
\begin{itemize}
\item[$(1)$]
$i_{n-1}+i_n<e$ and

\item[$(2)$]
the number of elements of $\big\{i_j-i_k\,|\,1\leq k<j\leq n\big\}$ is equal to $\frac{n(n-1)}{2}$,
\end{itemize}
then the minimum insdel distance of  $\mathcal{C}$ is equal to $2n-4$. In other words, $\mathcal{C}$ is optimal in the sense that it meets the
bound obtained in Theorem {\bf A}.
}
\vskip 0.1cm

As a consequence, we show that the  conditions $(1)$ and $(2)$ listed in Theorem {\bf B} are  not hard to achieve;
we explicitly construct an infinite family of optimal two-dimensional Reed-Somolom codes meeting the bound, as we show below.

\vskip 0.1cm
\textbf{Corollary C}~
{\it Let $p$ be a prime number and let $e>1$ be a positive integer. Let  $i_j=2^{j-1}$ for $1\leq j\leq n$
satisfying  $3\cdot 2^{n-2}<e$. Let $\theta$ be a primitive element in the finite field $\mathbb{F}_{p^e}$.
Then
$$\mathcal{C}=\Big\{\Big(\lambda+\mu\theta^{i_1}, \lambda+\mu\theta^{i_2},\cdots,\lambda+\mu\theta^{i_n}\Big)\,\Big{|}\,\lambda,\mu\in \mathbb{F}_{p^e}\Big\}$$
is a two-dimensional Reed-Solomon code of length $n$ over $\mathbb{F}_{p^e}$ whose minimum insdel distance is equal to
$d(\mathcal{C})=2n-4.$}
\vskip 0.1cm

This paper is organized as follows. In Section $2$, we recall some definitions and basic results about general linear codes,  insdel codes and Reed-Solomon codes.
In Section  $3$,  we give the proof  for    Theorem \textbf{A},   and in Section $4$, the proofs for Theorem  \textbf{B} and Corollary \textbf{C} are presented.
We conclude this paper with remarks on possible future works in Section $5$.
\section{Preliminaries}
Let $\mathbb{F}_q$ be a finite field with $q$ elements and
let $\mathbb{F}_q^n$ be the set of all vectors of length $n$ over $\mathbb{F}_q$.
A subspace $\mathcal{C}$ of $\mathbb{F}_q^n$ over  $\mathbb{F}_q$ is called a linear code of length $n$ over $\mathbb{F}_q$.
The Hamming distance between two vectors $\mathbf{x}, \mathbf{y}\in \mathbb{F}_q^n$,
which is defined to be the number of coordinates in which $\mathbf{x}$ and $\mathbf{y}$ differ, is denoted by $d_H(\mathbf{x},\mathbf{y})$.
The minimum Hamming distance $d_H(\mathcal{C})$ of a code $\mathcal{C}$ is
the smallest Hamming distance among all pairs of distinct codewords of $\mathcal{C}$.
The Hamming weight $w(\mathbf{x})$ of a vector $\mathbf{x}\in \mathbb{F}_q^n$ is the number of nonzero coordinates in $\mathbf{x}$.
It is well known that if $\mathcal{C}$ is a linear code, then the minimum Hamming distance $d_H(\mathcal{C})$ is the same as
the minimum Hamming weight of the nonzero codewords of $\mathcal{C}$.
A linear code $\mathcal{C}$ of length $n$,  dimension $k$ and minimum Hamming distance $d_H(\mathcal{C})$ over $\mathbb{F}_q$ is often
called a $q$-ary $[n, k,d_H(\mathcal{C})]$ code or, if $q$ is clear from the context, an $[n, k,d_H(\mathcal{C})]$ code.
It is well known that an $[n,k,d_H(\mathcal{C})]$ linear code $\mathcal{C}$ over $\mathbb{F}_q$ must obey the Singleton bound,
i.e., the code length $n$, dimension $k$ and minimum Hamming distance $d_H(\mathcal{C})$ satisfy
$$d_H(\mathcal{C})\leq n-k+1.$$
The $[n,k,d_H(\mathcal{C})]$ linear codes $\mathcal{C}$ over $\mathbb{F}_q$ meeting the Singleton bound  are called maximum distance separable code (MDS code for short).

In this paper, for linear codes over $\mathbb{F}_q$, we mainly consider the insdel distance used in high insertion and deletion noise regime.
We restate this definition as follows.

\begin{Definition}
For two vectors $\mathbf{a},\mathbf{b}\in \mathbb{F}_q^n$, the insdel distance $d(\mathbf{a},\mathbf{b})$ between   $\mathbf{a}$ and $\mathbf{b}$ is the minimum number of
insertions and deletions which are needed to transform $\mathbf{a}$ into $\mathbf{b}$.
It can be verified that $d(\mathbf{a},\mathbf{b})$ is indeed a metric on $\mathbb{F}_q^n$.
\end{Definition}

It has been shown that the insdel distance between any two vectors can be characterized via their longest common subsequences.

\begin{lem}\label{estimatevalue1}
\cite[Lemma 1]{dltx} Let $\mathbf{a},\mathbf{b}\in \mathbb{F}_q^n$. Then we have
$$d(\mathbf{a},\mathbf{b})=2n-2\ell,$$
where $\ell$ denotes the length of a longest common subsequence of $\mathbf{a}$ and $\mathbf{b}$.
\end{lem}

Lemma \ref{estimatevalue1} is useful in   calculating  the insdel distance of two vectors in $\mathbb{F}_q^n$.

Similar to the definition of  minimum Hamming distance of linear codes over $\mathbb{F}_q$,
we  give the definition of minimum insdel distance of a linear code over $\mathbb{F}_q$ below,
which is one of the most important parameters as it indicates the insdel error-correcting capability.

\begin{Definition}
An insdel linear code $\mathcal{C}$ of length $n$ is a linear subspace of $\mathbb{F}_q^n$
with minimum insdel distance being defined as
$$d(\mathcal{C})=\min_{\mathbf{c}_1,\mathbf{c}_2\in \mathcal{C}, \mathbf{c}_1\neq \mathbf{c}_2}\{d(\mathbf{c}_1,\mathbf{c}_2)\}.$$
\end{Definition}

An linear code $\mathcal{C}$ over $\mathbb{F}_q$ of length $n$, dimension $k$ and minimum insdel distance $d(\mathcal{C})$ is called an $[n,k,d(\mathcal{C})]$ insdel linear code over $\mathbb{F}_q$.
As we mentioned in the first section, an $[n,k,d(\mathcal{C})]$ insdel linear code $\mathcal{C}$ must obey the following Singleton-type bound.

\begin{Proposition}
(Singleton Bound \cite{hs}) Let $\mathcal{C}$ be an $[n,k,d(\mathcal{C})]$ insdel linear code over $\mathbb{F}_q$. Then
$$d(\mathcal{C})\leq 2n-2k+2.$$
\end{Proposition}

In the rest of this section we give the definition and some basic facts about Reed-Solomon codes.
Let $n>k$ be two positive integers.
Let $\mathbb{F}_q$ be a finite field with $q$ elements and choose $n$ distinct elements $\alpha_1, \alpha_2,\cdots,\alpha_n$ of
$\mathbb{F}_q$.
Denote $\mathbb{F}_q^{<k}[x]$ by the set of polynomials in $\mathbb{F}_q[x]$ of degree less that $k$.
For $1\leq k<n$,   the Reed-Solomon code of length $n$ and dimension $k$ with code locators $\{\alpha_1, \alpha_2,\cdots,\alpha_n\}$
is defined as
$$\textbf{RS}_{n,k}(\alpha)=\{(f(\alpha_1), f(\alpha_2),\cdots,f(\alpha_n))\,|\,f(x)\in \mathbb{F}_q^{<k}[x]\}.$$
Then $\textbf{RS}_{n,k}(\alpha)$ is an $[n,k,n-k+1]$ linear code over $\mathbb{F}_q$ with length $n\leq q$.
In particular, Reed-Solomon codes are MDS codes.

\section{Proof of Theorem A}
Let $\mathcal{C}$ be a linear code of length $n$ over $\mathbb{F}_q$. As before, the minimum Hamming distance of the linear code
$\mathcal{C}$ is denoted by $d_H(\mathcal{C})$; the minimum insdel distance of  $\mathcal{C}$ is denoted by
$d(\mathcal{C})$.
For two typical codewords $\mathbf{a},\mathbf{b}$ of $\mathcal{C}$,  the Hamming (resp. insdel) distance between $\mathbf{a}$ and $\mathbf{b}$ is denoted by
$d_H(\mathbf{a},\mathbf{b})$ (resp. $d(\mathbf{a},\mathbf{b})$). It follows from Lemma \ref{estimatevalue1} that the insdel distance between
$\mathbf{a}$ and $\mathbf{b}$ is less than or equal to  $2n$, i.e., $d(\mathbf{a},\mathbf{b})\leq 2n$.  In order to prove Theorem {\bf A},
we first need to improve this upper bound, as we show below.

\begin{lem}\label{estimatevalue2}
Let  $\mathbf{a}$ and $\mathbf{b}$ be two   vectors of length $n$ over $\mathbb{F}_q$.
Then we have
$$d(\mathbf{a},\mathbf{b})\leq 2d_H(\mathbf{a},\mathbf{b}).$$
\end{lem}
\begin{proof}
Let $\ell$ denote the length of a longest common subsequence of   $\mathbf{a}$ and $\mathbf{b}$.
Observe that $\ell\geq n-d_H(\mathbf{a},\mathbf{b})$ and then by Lemma \ref{estimatevalue1} we immediately  have
$$d(\mathbf{a},\mathbf{b})=2n-2\ell\leq 2n-2\big(n-d_H(\mathbf{a},\mathbf{b})\big)=2d_H(\mathbf{a},\mathbf{b}).$$
The lemma is proved.
\end{proof}

We have shown that the insdel distance of arbitrary two vectors in $\mathbb{F}_q^n$ is at most twice of their Hamming distance.
We next show that the same conclusion holds for the minimum insdel distance and the minimum Hamming distance of any linear code.
\begin{lem}\label{estimatevalue3}
Let $\mathcal{C}$ be a non-zero linear code of length $n$ over $\mathbb{F}_q$. We then have
$$d(\mathcal{C})\leq 2d_H(\mathcal{C}).$$
\end{lem}
\begin{proof}\label{estimatevalue4}
Choose two distinct codewords $\mathbf{a_0}$ and $\mathbf{b_0}$ of $\mathcal{C}$ such  that $d_H(\mathbf{a_0},\mathbf{b_0})=d_H(\mathcal{C})$.
Recall that
$d(\mathcal{C})=\min\{d(\mathbf{a},\mathbf{b})\,|\,\mathbf{a},\mathbf{b}\in \mathcal{C}, \mathbf{a}\neq \mathbf{b}\}$.
Thus by Lemma \ref{estimatevalue2} we have
$$d(\mathcal{C})\leq d(\mathbf{a_0},\mathbf{b_0})\leq 2d_H(\mathbf{a_0},\mathbf{b_0})=2d_H(\mathcal{C}).$$
We are done.
\end{proof}

\begin{Remark}\label{remark}{\rm
Note, by the classical Singleton bound of an $[n,k,d_H(\mathcal{C})]$ linear code $\mathcal{C}$,
that   $d_H(\mathcal{C})\leq n-k+1$.
We therefore conclude from Lemma \ref{estimatevalue3} that the minimum insdel distance of an $[n,k]$ linear code
$\mathcal{C}$ must be less than or equal to
$2(n-k+1)$, i.e.,
$$d(\mathcal{C})\leq 2n-2k+2.$$
This  Singleton bound for an insdel code was exhibited in \cite{hs}.
We also note by Lemma \ref{estimatevalue3} that if the minimum Hamming distance of $\mathcal{C}$ is at most $n-k$, then
$d(\mathcal{C})\leq2d_H(\mathcal{C})=2n-2k$, proving Theorem {\bf A} in this special case.
By virtue of this fact, for the goal of completing the proof Theorem {\bf A},
we only need to restrict ourself to the case where $\mathcal{C}$ is an MDS code.}
\end{Remark}

The following lemma gives the desired result, which is a crucial step in the process of proving our Theorem {\bf A}.
\begin{lem}\label{MDS}
Let $\mathcal{C}$ be a linear $[n,k]$ MDS code over $\mathbb{F}_q$ with $n>k\geq 2$. Then
$$d(\mathcal{C})\leq 2n-2k.$$
\end{lem}
\begin{proof}
We use a characterization of MDS codes to complete the proof:
An $[n,k,d_H]$ linear code $\mathcal{C}$ over  $\mathbb{F}_q$ is MDS
if and only if $\mathcal{C}$ has a
minimum weight codeword in any $d_H$ coordinates (see \cite[Chap. 11, Theorem 4]{ms}).
Now $\mathcal{C}$ is an $[n,k]$ MDS code over $\mathbb{F}_q$ with  $n>k\geq 2$, which gives that $\mathcal{C}$ has a
minimum weight codeword in any $n-k+1$ coordinates. First we choose the last  $n-k+1$ coordinates to be non-zero, and we suppose further that the
$k$th coordinate is equal to $1$,  i.e.,


$$\mathbf{a}=\big(\underbrace{0,\cdots,0}_{k-1},1,*,\cdots,*\big)$$
where $*$  denotes some non-zero elements of $\mathbb{F}_q$. It follows that suitable non-zero elements $*$ of
$\mathbb{F}_q$ can be found such that
$\mathbf{a}$ is a codeword of $\mathcal{C}$.
Likewise, suitable non-zero elements $\star$ of $\mathbb{F}_q$ can be found such that
$$\mathbf{b}=(\star,\underbrace{0,\cdots,0}_{k-1},1,\star,\cdots,\star)$$
is also a codeword of $\mathcal{C}$. Now we have two distinct codewords $\mathbf{a}\in \mathcal{C}$ and $\mathbf{b}\in \mathcal{C}$.
Thus the length $\ell$ of a longest common subsequence of $\mathbf{a}$ and $\mathbf{b}$ satisfies $\ell\geq k$.
Therefore,  by Lemma \ref{estimatevalue1} we get that
$$d(\mathcal{C})\leq d(\mathbf{a},\mathbf{b})=2n-2\ell\leq 2n-2k.$$
We are done.
\end{proof}

By Remark \ref{remark} and Lemma \ref{MDS}, we immediately arrive at Theorem {\bf A},
which says that the minimum insdel distance of any $[n,k]$ linear code with $n>k\geq2$ is at most $2n-2k$.

\section{Proofs of Theorem B and  Corollary C}
The primary goal of this section is to present a proof for Theorem {\bf B},
which gives a sufficient condition to guarantee a  two-dimensional
Reed-Solomon code  of length $n$ over $\mathbb{F}_q$ to have minimum insdel distance   $2n-4$.
Such codes  are optimal in the sense that they meet the upper bound
obtained in Theorem {\bf A}.

We first fix some notation.
Let $\mathbb{F}_q$ be a finite field with $q=p^e$ elements,
where $p$ is a prime number and $e>1$ is a positive integer.
Let $\mathbb{F}_q^*$ be the multiplicative group of $\mathbb{F}_q$,
and let $\theta$ be a primitive element of   $\mathbb{F}_q$, i.e., $\theta$ generates the cyclic group $\mathbb{F}_q^*$.
Then $e$ is the degree of the minimal polynomial of $\theta$ over the prime field $\mathbb{F}_p$.

Let $i_1, i_2, \cdots,i_n$ be $n$ integers such that $0\leq i_1<i_2< \cdots<i_n\leq q-2$, where $n\geq 3$.
Let $\mathcal{C}$ be a two-dimensional Reed-Solomon code of length $n$ over $\mathbb{F}_q$
with code locators $\{\theta^{i_1}, \theta^{i_2},\cdots,\theta^{i_n}\}$, i.e.,
\begin{equation}\label{code}
\mathcal{C}=\Big\{\Big(\lambda+\mu\theta^{i_1}, \lambda+\mu\theta^{i_2},\cdots,\lambda+\mu\theta^{i_n}\Big)\,\Big{|}\,\lambda,\mu\in \mathbb{F}_q\Big\}.
\end{equation}
Our goal is, therefore, converting to find suitable numbers $i_1,i_2,\cdots, i_n$ such that $\mathcal{C}$ has minimum insdel distance $2n-4$.
For this purpose,
let
$$D=\big\{i_j-i_k\,\,|\,\,1\leq k<j\leq n\big\},$$
then $$D\subseteq\big\{1,2,\cdots,q-2\big\}.$$

We are now in a position to prove Theorem {\bf B}.

\begin{proof}
Suppose $\mathcal{C}$ is a two-dimensional Reed-Solomon code of length $n$ over $\mathbb{F}_q$ as given in (\ref{code}).
According to Theorem {\bf A} we have that $d(\mathcal{C})\leq 2n-4$. Thus it remains to show that  $d(\mathcal{C})\geq 2n-4$.
To this end, by  virtue of Lemma \ref{estimatevalue1},  it is enough to  prove the following claim:

\textbf{Claim}: $\ell(\mathbf{a},\mathbf{b})\leq 2$ for any distinct codewords $\mathbf{a},\mathbf{b}\in\mathcal{ C}$,
where $\ell(\mathbf{a}, \mathbf{b})$ is the length of a longest common subsequence of $\mathbf{a}$ and $\mathbf{b}$.

Now we consider $7$ cases separately to investigate the number $\ell(\mathbf{a}, \mathbf{b})$.


Case $1$: Either $\mathbf{a}$ or $\mathbf{b}$ is the zero vector. It is trivial to see that  $\ell(\mathbf{a}, \mathbf{b})\leq2$ because $\mathcal{C}$ is an MDS code with parameters $[n,2,n-1]$.

Henceforth, we can assume that both $\mathbf{a}$ and $\mathbf{b}$ are non-zero.

Case $2$: Assume that
$$\mathbf{a}=\big(\lambda,\lambda,\cdots,\lambda\big), ~~\mathbf{b}=\big(\mu,\mu,\cdots,\mu\big),$$
where $\lambda,\mu$ are  elements of $\mathbb{F}_q$ with $\lambda\neq \mu$. Then it is easy to see that $\ell(\mathbf{a}, \mathbf{b})\leq2$.

Case $3$: Assume that
$$\mathbf{a}=(\lambda,\lambda,\cdots,\lambda), ~~\mathbf{b}=(\mu\theta^{i_1},\mu\theta^{i_2},\cdots,\mu\theta^{i_n}),$$
where $\lambda,\mu\in \mathbb{F}_q^*$.
Suppose otherwise that there exist three integers $r_1,r_2,r_3$ satisfying $1\leq r_1<r_2<r_3\leq n$ such that
\begin{equation*}
\begin{cases}
\lambda=\mu\theta^{i_{r_1}},\\
\lambda=\mu\theta^{i_{r_2}},\\
\lambda=\mu\theta^{i_{r_3}}.
\end{cases}
\end{equation*}
This gives
$$\theta^{i_{r_1}}=\theta^{i_{r_2}}=\theta^{i_{r_3}}=\frac{\lambda}{\mu}.$$
Thus
$$\theta^{i_{r_2}-i_{r_1}}=\theta^{i_{r_3}-i_{r_1}}=1.$$
Recall that $i_{r_2}>i_{r_1}$,  $i_{r_3}>i_{r_1}$ and $0\leq i_{r_j}\leq q-2$ for $j=1,2,3$.    We then have
$$i_{r_2}-i_{r_1}=i_{r_3}-i_{r_1}.$$
This contradicts to the condition $(2)$ in the theorem. It follows that $\ell(\mathbf{a},\mathbf{b})\leq 2$.

Case $4$: Assume that
$$\mathbf{a}=(\lambda,\lambda,\cdots,\lambda), ~~ \mathbf{b}=(\lambda_1+\mu_1\theta^{i_1},\lambda_1+\mu_1\theta^{i_2},\cdots,\lambda_1+\mu_1\theta^{i_n}),$$
where $\lambda,\lambda_1,\mu_1\in \mathbb{F}_q^*$.
Suppose otherwise that there exist three integers $r_1,r_2,r_3$ satisfying $1\leq r_1<r_2<r_3\leq n$ such that
\begin{equation*}
\begin{cases}
\lambda=\lambda_1+\mu_1\theta^{i_{r_1}},\\
\lambda=\lambda_1+\mu_1\theta^{i_{r_2}},\\
\lambda=\lambda_1+\mu_1\theta^{i_{r_3}}.
\end{cases}
\end{equation*}
This leads to
$$\theta^{i_{r_1}}=\theta^{i_{r_2}}=\theta^{i_{r_3}}=\frac{\lambda-\lambda_1}{\mu_1}.$$
Thus
$$\theta^{i_{r_2}-i_{r_1}}=\theta^{i_{r_3}-i_{r_1}}=1,$$
which gives
$$i_{r_2}-i_{r_1}=i_{r_3}-i_{r_1}.$$
This  contradicts to the condition $(2)$ in the theorem again. It follows that $\ell(\mathbf{a},\mathbf{b})\leq 2$.

Case $5$: Assume that
$$\mathbf{a}=(\lambda\theta^{i_1},\lambda\theta^{i_2},\cdots,\lambda\theta^{i_n}), ~~ \mathbf{b}=(\mu\theta^{i_1},\mu\theta^{i_2},\cdots,\mu\theta^{i_n}),$$
where $\lambda,\mu\in \mathbb{F}_q^*$ with $\lambda\neq \mu$.
Suppose otherwise that there exist six integers $k_1, k_2, k_3, r_1,r_2,r_3$ satisfying $1\leq k_1< k_2< k_3\leq n, ~1\leq r_1<r_2<r_3\leq n$ such that
\begin{equation*}
\begin{cases}
\lambda\theta^{i_{k_1}}=\mu\theta^{i_{r_1}},\\
\lambda\theta^{i_{k_2}}=\mu\theta^{i_{r_2}},\\
\lambda\theta^{i_{k_3}}=\mu\theta^{i_{r_3}}.
\end{cases}
\end{equation*}
This gives
$$\theta^{i_{k_1}-i_{r_1}}=\theta^{i_{k_2}-i_{r_2}}=\theta^{i_{k_3}-i_{r_3}}=\frac{\mu}{\lambda},$$
which implies that
\begin{equation*}
\begin{cases}
i_{k_1}-i_{r_1}\equiv i_{k_2}-i_{r_2}(\textrm{mod}~q-1),\\
i_{k_1}-i_{r_1}\equiv i_{k_3}-i_{r_3}(\textrm{mod}~q-1).
\end{cases}
\end{equation*}
This is equivalent to
\begin{equation*}
\begin{cases}
i_{k_2}-i_{k_1}\equiv i_{r_2}-i_{r_1}(\textrm{mod}~q-1),\\
i_{k_3}-i_{k_1}\equiv i_{r_3}-i_{r_1}(\textrm{mod}~q-1).
\end{cases}
\end{equation*}
Thus we have
\begin{equation*}
\begin{cases}
i_{k_2}-i_{k_1}= i_{r_2}-i_{r_1},\\
i_{k_3}-i_{k_1}= i_{r_3}-i_{r_1}.
\end{cases}
\end{equation*}
This is a contradiction again. It follows that $\ell(\mathbf{a},\mathbf{b})\leq 2$.

Case $6$: Assume that
$$\mathbf{a}=(\lambda\theta^{i_1},\lambda\theta^{i_2},\cdots,\lambda\theta^{i_n}),~~
\mathbf{b}=(\lambda_1+\mu_1\theta^{i_1},\lambda_1+\mu_1\theta^{i_2},\cdots,\lambda_1+\mu_1\theta^{i_n}),$$
where $\lambda,\lambda_1,\mu_1\in \mathbb{F}_q^*$.
Suppose otherwise that there exist six integers $k_1, k_2, k_3, r_1,r_2,r_3$ satisfying $1\leq k_1< k_2< k_3\leq n,~ 1\leq r_1<r_2<r_3\leq n$ such that
\begin{equation*}
\begin{cases}
\lambda\theta^{i_{k_1}}=\lambda_1+\mu_1\theta^{i_{r_1}},\\
\lambda\theta^{i_{k_2}}=\lambda_1+\mu_1\theta^{i_{r_2}},\\
\lambda\theta^{i_{k_3}}=\lambda_1+\mu_1\theta^{i_{r_3}}.
\end{cases}
\end{equation*}
In the matrix version, that is  equivalent to saying that
$$
\begin{pmatrix}
\theta^{i_{k_1}} & \theta^{i_{r_1}} & 1\\
\theta^{i_{k_2}} & \theta^{i_{r_2}} & 1\\
\theta^{i_{k_3}} & \theta^{i_{r_3}} & 1
\end{pmatrix}
\begin{pmatrix}
\lambda \\ -\mu_1 \\ -\lambda_1
\end{pmatrix}=\mathbf{0}.
$$
Since
$$
\begin{pmatrix}
\lambda \\ -\mu_1 \\ -\lambda_1
\end{pmatrix}\neq\mathbf{0},
$$
we have that the determinant of the coefficient matrix is zero, i.e.,
$$
\begin{vmatrix}
\theta^{i_{k_1}} & \theta^{i_{r_1}} & 1\\
\theta^{i_{k_2}} & \theta^{i_{r_2}} & 1\\
\theta^{i_{k_3}} & \theta^{i_{r_3}} & 1
\end{vmatrix}=0.
$$
Now suppose $x$ is an indeterminate  over the finite field $\mathbb{F}_q$.  Then we have the polynomial
$$f(x)=
\begin{vmatrix}
x^{i_{k_1}} & x^{i_{r_1}} & 1\\
x^{i_{k_2}} & x^{i_{r_2}} & 1\\
x^{i_{k_3}} & x^{i_{r_3}} & 1
\end{vmatrix}.
$$
It is clear that  $f(\theta)=0$. On the other hand, we can expand the determinant to have
\begin{eqnarray*}
f(x)&=&
\begin{vmatrix}
x^{i_{k_1}} & x^{i_{r_1}} & 1\\
x^{i_{k_2}} & x^{i_{r_2}} & 1\\
x^{i_{k_3}} & x^{i_{r_3}} & 1
\end{vmatrix}\\
&=& \begin{vmatrix}
x^{i_{k_1}}-x^{i_{k_3}} & x^{i_{r_1}}-x^{i_{r_3}} & 0\\
x^{i_{k_2}}-x^{i_{k_3}} & x^{i_{r_2}}-x^{i_{r_3}} & 0\\
x^{i_{k_3}} & x^{i_{r_3}} & 1
\end{vmatrix}\\
&=&(x^{i_{k_1}}-x^{i_{k_3}})(x^{i_{r_2}}-x^{i_{r_3}})-(x^{i_{k_2}}-x^{i_{k_3}})(x^{i_{r_1}}-x^{i_{r_3}})\\
&=& x^{i_{k_1}+i_{r_2}}+x^{i_{k_2}+i_{r_3}}+x^{i_{k_3}+i_{r_1}}-x^{i_{k_1}+i_{r_3}}-x^{i_{k_2}+i_{r_1}}-x^{i_{k_3}+i_{r_2}}.
\end{eqnarray*}
Observing  that $k_1< k_2< k_3$ and $r_1<r_2<r_3$, we have
$$i_{k_1}+i_{r_2}<\min\{i_{k_2}+i_{r_3},i_{k_1}+i_{r_3}\},~~
i_{k_2}+i_{r_1}<\min\{i_{k_3}+i_{r_1},i_{k_3}+i_{r_2}\}.$$
Thus the term of the minimum degree in the polynomial $f(x)$ is $x^{i_{k_1}+i_{r_2}}$ or $-x^{i_{k_2}+i_{r_1}}$.
As $i_{k_1}+i_{r_2}\neq i_{k_2}+i_{r_1}$, otherwise we would have $i_{k_2}-i_{k_1}=i_{r_2}-i_{r_1}$, a contradiction,
we conclude that $f(x)\neq 0$.

Since $f(x)$ is a non-zero polynomial, $f(\theta)=0$ and $e$ is the degree of the minimal polynomial of $\theta$ over $\mathbb{F}_p$,
we obtain that $e\leq\deg(f(x))$.
Thus it is easy to see that the degree $\deg(f(x))$ of $f(x)$ satisfies
$$e\leq\deg(f(x))\leq i_{n-1}+i_n< e.$$
The last inequality is from our condition $(1)$.
This is a contradiction, and we conclude the proof of $\ell(\mathbf{a},\mathbf{b})\leq 2$.

Case $7$: Assume that
$$\mathbf{a}=(\lambda_2+\mu_2\theta^{i_1},\lambda_2+\mu_2\theta^{i_2},\cdots,\lambda_2+\mu_2\theta^{i_n}),~~
\mathbf{b}=(\lambda_1+\mu_1\theta^{i_1},\lambda_1+\mu_1\theta^{i_2},\cdots,\lambda_1+\mu_1\theta^{i_n}),$$
where $\lambda_2,\mu_2,\lambda_1,\mu_1\in \mathbb{F}_q^*$.
Suppose otherwise that there exist six integers $k_1, k_2, k_3, r_1,r_2,r_3$ satisfying $1\leq k_1< k_2< k_3\leq n, 1\leq r_1<r_2<r_3\leq n$ such that
\begin{equation*}
\begin{cases}
\lambda_2+\mu_2\theta^{i_{k_1}}=\lambda_1+\mu_1\theta^{i_{r_1}},\\
\lambda_2+\mu_2\theta^{i_{k_2}}=\lambda_1+\mu_1\theta^{i_{r_2}},\\
\lambda_2+\mu_2\theta^{i_{k_3}}=\lambda_1+\mu_1\theta^{i_{r_3}}.
\end{cases}
\end{equation*}
We then have
$$
\begin{pmatrix}
\theta^{i_{k_1}} & \theta^{i_{r_1}} & 1\\
\theta^{i_{k_2}} & \theta^{i_{r_2}} & 1\\
\theta^{i_{k_3}} & \theta^{i_{r_3}} & 1
\end{pmatrix}
\begin{pmatrix}
\mu_2 \\ -\mu_1 \\ \lambda_2-\lambda_1
\end{pmatrix}=\mathbf{0}.
$$
Since
$$
\begin{pmatrix}
\mu_2 \\ -\mu_1 \\ \lambda_2-\lambda_1
\end{pmatrix}\neq\mathbf{0},
$$
we have that
$$
\begin{vmatrix}
\theta^{i_{k_1}} & \theta^{i_{r_1}} & 1\\
\theta^{i_{k_2}} & \theta^{i_{r_2}} & 1\\
\theta^{i_{k_3}} & \theta^{i_{r_3}} & 1
\end{vmatrix}=0.
$$
Using the same discussion as in the Case $6$, we get a contradiction and conclude the proof of $\ell(\mathbf{a},\mathbf{b})\leq 2$.

Based on the above $7$ cases, we have established the claim. Then by Lemma \ref{estimatevalue1} we obtain $d(\mathcal{C})\leq 2n-4$,
which forces $d(\mathcal{C})= 2n-4$ by Theorem {\bf A}. This completes the proof.
\end{proof}

With Theorem {\bf B}, we can prove Corollary {\bf C} easily, which generates an infinite family of optimal two-dimensional Reed-Somolom codes meeting the bound in Theorem {\bf A}.

\begin{proof}
It is enough to check the conditions (1) and (2) in Theorem {\bf B} are all satisfied. Condition (1) in Theorem {\bf B} is satisfied since
$i_{n-1}+i_n=3\cdot 2^{n-2}<e$. Next, let us compute $|D|$. Suppose that
$$2^i-2^j=2^s-2^t~\big(0\leq i,j,s,t\leq n-1,~ j<i,~ t<s\big).$$
If $j\neq t$, then without loss of generality, assume that $j>t$. This gives $2^{i-t}-2^{j-t}=2^{s-t}-1$, a contradiction,
which shows that $j=t$ and  $i=s$. Therefore $|D|=\frac{1}{2}n(n-1)$. We have shown that Condition (2) in Theorem {\bf B} is satisfied.
By Theorem {\bf B}, we conclude that the code $\mathcal{C}$  is a two-dimensional Reed-Solomon code  of length $n$  with  minimum insdel distance $d(\mathcal{C})=2n-4$.
\end{proof}

\section{Conclusion and future work}
In this paper, we showed that if
$\mathcal{C}$ is an $[n,k]$ linear code over $\mathbb{F}_q$ with $n>k\geq2$, then the minimum insdel distance of $\mathcal{C}$
is at most $2n-2k$ (see Theorem {\bf A}). This result significantly improves the previously known results in \cite{dltx} and \cite{hs} as we mentioned in the Introduction section.
We gave   a sufficient condition under which a  two-dimensional
Reed-Solomon code  of length $n$ over $\mathbb{F}_q$   has minimum insdel distance   $2n-4$ (see Theorem {\bf B});
as a corollary, we showed that the conditions listed in Theorem {\bf B} are easy to achieve (see Corollary {\bf C}).
Consequently, we have explicitly constructed an infinite family of optimal two-dimensional Reed-Somolom codes meeting the bound in Theorem {\bf A}.
Comparing with \cite{dltx}, our methods are more direct and easy to understand.

A possible
direction for future  work is to find non-MDS codes that meet the bound exhibited in Theorem {\bf A}; if this can be done,
it may lead us to know more about insertion-deletion metric. Apart from this problem,
there could be many other interesting problems associated with
insertion-deletion codes. For instance, it would be interesting to establish other bounds with respect to the insertion-deletion metric and give some optimal constructions.



\end{document}